\documentclass[12 pt]{amsart}
\usepackage{amssymb,latexsym,amsmath,amscd,amsthm,graphicx, color}
\usepackage[all]{xy}
\usepackage{pgf,tikz}
\usepackage{mathrsfs}
\usepackage{cite}
\usetikzlibrary{arrows}
\definecolor{uuuuuu}{rgb}{0.26666666666666666,0.26666666666666666,0.26666666666666666}
\definecolor{draff}{rgb}{0.49019607843137253,0.49019607843137253,1.}
\definecolor{faffed}{rgb}{1.,0.,0.}

\definecolor{uuuuuu}{rgb}{0.26666666666666666,0.26666666666666666,0.26666666666666666}
\definecolor{qqwuqq}{rgb}{0.,0.39215686274509803,0.}
\definecolor{zzttqq}{rgb}{0.6,0.2,0.}
\definecolor{draff}{rgb}{0.49019607843137253,0.49019607843137253,1.}
\definecolor{qqqqff}{rgb}{0.,0.,1.}
\definecolor{cqcqcq}{rgb}{0.7529411764705882,0.7529411764705882,0.7529411764705882}

\theoremstyle{plain}

\newtheorem{lemma}[subsection]{Lemma}

\theoremstyle{definition}
\newtheorem{defi}[subsection]{Definition}

\newtheorem{example}[subsection]{Example}

\newtheorem{remark}[subsection]{Remark}

\newtheorem{note}[subsection]{Note}

\newcommand{\uu}{\cup}% union
\newcommand{\UU}{\bigcup}% big union
\newcommand{\ci}{\subseteq}% contained in with equality
\newcommand{\sci}{\subset}% strictly contained in
\newcommand{\es}{\emptyset}% the empty set
\newcommand{\set}[1]{\{#1\}}% set
\newcommand{\ga}{\alpha}
\newcommand{\gb}{\beta}

\newcommand{\gd}{\delta}
\renewcommand{\gg}{\gamma}% old use >>

\newcommand{\gs}{\sigma}

\newcommand{\tit}{\textit}% text italic

\newcommand{\D}[1]{\mathbb{#1}}% Doubled - blackboard bold - only caps, uas as \D{A}
\newcommand{\te}{\text}% same as \mathrm command.

\begin{document}

To appear, Southeast Asian Bulletin of Mathematics
\title{An algorithm to compute CVTs for finitely generated Cantor distributions}

\author{Carl P. Dettmann}
\address{University of Bristol\\
School of Mathematics\\
University Walk\\
Bristol BS8 1TW\\
UK.}
\email{Carl.Dettmann@bris.ac.uk}

\author{ Mrinal Kanti Roychowdhury}
\address{School of Mathematical and Statistical Sciences\\
University of Texas Rio Grande Valley\\
1201 West University Drive\\
Edinburg, TX 78539-2999, USA.}
\email{mrinal.roychowdhury@utrgv.edu}
\thanks{The research of the first author was supported by the Engineering and Physical Sciences Research Council (EPSRC) Grant EP/N002458/1, and that of the second author was supported by U.S. National Security Agency (NSA) Grant H98230-14-1-0320}

\thanks{}

\date{}
\maketitle

\pagestyle{myheadings}\markboth{Carl P. Dettmann and Mrinal Kanti Roychowdhury}{An algorithm to compute CVTs for finitely generated Cantor distributions}

\begin{abstract}
Centroidal Voronoi tessellations (CVTs) are Voronoi tessellations of a region such that the generating points of the tessellations are also the centroids of the corresponding Voronoi regions with respect to a given probability measure. CVT is a fundamental notion that has a wide spectrum of applications in computational science and engineering. In this paper, an algorithm is given to obtain the CVTs with $n$-generators to level $m$, for any positive integers $m$ and  $n$, of any  Cantor set generated by a pair of self-similar mappings given by $S_1(x)=r_1x$ and $S_2(x)=r_2x+(1-r_2)$ for $x\in \mathbb R$, where $r_1, r_2>0$ and $r_1+r_2<1$, with respect to any probability distribution $P$ such that $P=p_1 P\circ S_1^{-1}+p_2 P\circ S_2^{-1}$, where $p_1, p_2>0$ and $p_1+p_2=1$.
\end{abstract}

\section{Introduction}
 Let $\D R^d$ denote the $d$-dimensional Euclidean space, $\|\cdot\|$ denote the Euclidean norm on $\D R^d$ for any $d\geq 1$, and $P$ be a Borel probability measure on $\D R^d$. Given a finite set $\ga\sci \D R^d$, the Voronoi region generated by $a\in \ga$ is defined by
\[W(a|\ga)=\set{x \in \D R^d : \|x-a\|=\min_{b \in \ga}\|x-b\|}\]
i.e., the Voronoi region generated by $a\in \ga$ is the set of all elements in $\D R^d$ which are closest to $a \in \ga$, and the set $\set{W(a|\ga) : a \in \ga}$ is called the \tit{Voronoi diagram} or \tit{Voronoi tessellation} of $\D R^d$ with respect to $\ga$. A Borel measurable partition $\set{A_a : a \in \ga}$ of $\D R^d$  is called a \tit{Voronoi partition} of $\D R^d$ with respect to $\ga$ (and $P$) if $P$-almost surely
$A_a \sci W(a|\ga)$ for every $a \in \ga$.  Given a Voronoi tessellation $\set{M_i}_{i=1}^k$ generated by a set of points $\set{z_i}_{i=1}^k$ (called \tit{sites} or \tit{generators}), the mass centroid  $c_i$ of $M_i$ with respect to the probability measure $P$ is given by
\begin{align*}
c_i=\frac{1}{P(M_i)}\int_{M_i} x dP(x)=\frac{\int_{M_i} x dP(x)}{\int_{M_i} dP(x)}.
\end{align*}
The Voronoi tessellation is called the \tit{centroidal Voronoi tessellation} (CVT) if $z_i=c_i$ for $i=1, 2, \cdots, k$, that is, if the generators are also the centroids of the corresponding Voronoi regions. It is interesting to notice that CVTs are not necessarily unique for a fixed probability measure and the number of generators, i.e., it is possible to have two or more different CVTs for a fixed probability measure and the number of generators (see \cite {DFG} for absolutely continuous probability measure, and see \cite{R1} for singular continuous probability measure). CVT generates an evenly-spaced distribution of sites in the domain with respect to a given probability measure and is therefore very useful in many fields, such as optimal quantization, clustering, data compression, optimal mesh generation, cellular biology, optimal quadrature, coverage control and geographical optimization (see \cite{DFG, OBSC} for more details). Besides, it has  applications in energy efficient distribution of base stations in a cellular network (see \cite{HCHSVH, KKR, S}). In both geographical and cellular applications
the distribution of users is highly complex and often modelled by a fractal (see \cite{ABDHW, LZSC}).

Let $\ga\sci \D R^d$ be a finite set. Then, the \tit{cost} or \tit{distortion} error for $P$ with respect to the set $\ga$, denoted by $V(P;\ga)$, is given by
\[V(P; \ga):=\int \min_{a \in \ga} \|x-a\|^2 dP(x). \]
Notice that if $\ga:=\set{a_1, a_2, \cdots, a_k}$, then
\[V(P; \ga):=\sum_{j=1}^k\int_{A_j} \|x-a_j\|^2 dP(x),\]
where $A_j$ is the Voronoi region of $a_j$, i.e., $a_j=W(a_j|\ga)$ for all $1\leq j\leq k$.
Then, the \tit{$n$th quantization error} for $P$, denoted by $V_n:=V_n(P)$, is defined by
\[V_n(P)=\te{inf}\set{V(P; \ga) : \ga \sci \D R^d, \, \te{card}(\ga) \leq n}.\]
If $\int \| x\|^2 dP(x)<\infty$, then there is some set $\ga$ for which the infimum is achieved (see \cite{ GKL, GL, GL1}). Such a set $\ga$ for which the infimum occurs and contains no more than $n$ points is called an \tit{optimal set of $n$-means}.  It is well-known that for a continuous Borel probability measure an optimal set of $n$-means always contains $n$ elements (see \cite{GL1}).
To know more details about quantization, one is referred to \cite{AW, GG, GL1, GN}. To see some work in the direction of optimal sets of $n$-means, one is refereed to \cite{DR, GL2, R2, R3, R4, RR}.
For a Borel probability measure $P$ on $\D R^d$, an optimal set of $n$-means forms a CVT with $n$-means ($n$-generators) of $\D R^d$; however, the converse is not true in general (see \cite{DFG, R1}). A CVT with $n$-means is called an \tit{optimal CVT with $n$-means} if the generators of the CVT form an optimal set of $n$-means with respect to the probability distribution $P$.

There are many applications of quantization of measure~\cite{GL1,I}.  As an example, suppose we want to locate
cellular phone towers or smaller access points so as to minimize the power consumption which increases as the
square of distance to the mobile users.  These users are however distributed in a very non-uniform manner
often modelled by a fractal~\cite{SD}, thus we need to minimize quantization error over this fractal measure.

Let $X$ be a nonempty compact subset of $\D R^d$; sometimes one can take $X=\D R^d$. A transformation $S: X \to X$ on a metric space $(X, d)$ is called a \tit{contractive} or a \tit{contraction mapping} if there is a constant $0<c<1$ such that $d(S(x), S(y))\leq c d(x, y)$ for all $x, y \in X$. On the other hand, $S$ is called a \tit{similarity mapping} or a \tit{similitude} if there exists a constant $s>0$ such that $d(S(x), S(y))=s d(x, y)$ for all $x, y\in X$. Here $s$ is called the \tit{similarity ratio} or the \tit{similarity constant} of the similarity mapping $S$. For any $N\geq 2$, an iterated function system (IFS) on $X$ is a collection of contraction mappings $S_1, S_2, \cdots, S_N$ on $X$. It is well-known that if the $S_j$ are contractions on $X$, then there exists a unique nonempty compact subset $J$ of $X$ such that
\[J=\UU_{j=1}^N S_j(J),\]
(see \cite{H, F}). We call $J$  \tit{the invariant set} or \tit{the attractor} or \tit{the limit set} of the IFS.

The open set condition (OSC) is the statement that there exists an open set $V\subset X$ such that $\UU_{j=1}^N
S_j(V)\subseteq V$ and the union is disjoint, informally, that the images of $J$ under the contractions do not overlap
too much.  The OSC is an important condition in many proofs, for example, regarding Hausdorff dimension~\cite{BVH05}.

If we associate the IFS with a probability vector $p=(p_1, p_2, \cdots, p_N)$, with $p_j>0$ for all $1\leq j\leq N$, then there exists a unique Borel probability measure $\mu_p$ on $\D R^d$ with supp$(\mu_p)=J$ such that $\mu_p=\sum_{j=1}^N p_j\mu_p\circ S_j^{-1}$, where $\mu_p\circ S_j^{-1}$ denotes the image measure of $\mu_p$ with respect to
$S_j$ for $1\leq j\leq N$, i.e.,
$ \mu_p(A)=\sum_{j=1}^Np_j\mu_p\circ S_j^{-1}(A),$
for all Borel sets $A \ci X$. We call $\mu_p$ the \tit{invariant measure} of the IFS associated with the probability vector $p$ (see \cite{H, F}).

 Let $C$ be the Cantor set generated by the two contractive similarity mappings $S_1$ and $S_2$ on $\D R$ such that $S_1(x)=r x$ and $ S_2 (x)=r x +(1-r)$ for all $x \in \D R$, where $0<r<\frac 1 2$. Let $P=\frac 1 2 P\circ S_1^{-1}+\frac 12 P\circ S_2^{-1}$. Then, $P$ is a singular continuous probability measure on $\D R$ with support the Cantor set $C$. If $r=\frac 13$, then in \cite{GL2}, Graf and Luschgy gave a formula to determine the optimal sets of $n$-means for the probability distribution $P$ for any $n\geq 2$. In \cite{R5},  L. Roychowdhury gave an induction formula for $n\geq 2$, to obtain the optimal sets of $n$-means for the Cantor distribution $P$ given by $P=\frac 14 P\circ S_1^{-1}+\frac 34 P\circ S_2^{-1}$ with support the Cantor set generated by the two mappings $S_1$ and $S_2$, where $S_1(x)=\frac 14 x$ and $S_2(x)=\frac 12 x+\frac 12$ for all $x \in \D R$. In \cite{R1}, the second author  gave a formula to determine the CVTs with $n$-means, $n\geq 2$, of the Cantor set generated by $S_1(x)=rx$ and $S_2(x)=rx+(1-r)$, $x\in \D R$, for any $r$ in the range  $0.4364590141\leq r\leq 0.4512271429$, associated with the probability distribution $P=\frac 12 P\circ S_1^{-1}+\frac 12 P\circ S_2^{-1}$.

There is no general formula to obtain the CVTs of any Cantor set generated by any two contractive similarity mappings $S_1$ and $S_2$ on $\D R$ such that $S_1(x)=r_1 x$ and $ S_2 (x)=r_2 x +(1-r_2)$ for all $x \in \D R$, where $r_1, r_2>0$ and $r_1+r_2<1$, supported by any probability distribution $P$ given by $P=p_1 P\circ S_1^{-1}+p_2 P\circ S_2^{-1}$, where $p_1, p_2>0$ and $p_1+p_2=1$. In this paper, we give an algorithm to obtain the CVTs with $n$-means at level $m$ of any Cantor set for any $n\geq 1$ supported by any probability distribution $P$ given by $P=p_1 P\circ S_1^{-1}+p_2 P\circ S_2^{-1}$.  Here, level $m$ refers to the granularity of the partition in terms of the natural partition of the Cantor set into cylinder sets.

We also give several examples and obtain the CVTs implementing the algorithm, as well as providing numerical simulations for the first nontrivial case $n=3$ in the uniform symmetric Cantor set case that illustrate the intricacy of the general problem.  We see that the optimal CVT occurs at a low level $m$ apart from the vicinity of $r=1/2$.  This can be understood in that placing the boundary between partition elements within the Cantor set (rather than at a significant gap) leads to unnecessary cost in the quantization error.  At $r=1/2$ the measure becomes Lebesgue on the unit interval, so the partition boundaries occur at $1/3$ and $2/3$, which are not at finite level with respect to the natural (dyadic) partition.

The algorithm in this paper can be extended to obtain the CVTs for any singular continuous probability measure supported by the limit set generated by a finite number of contractive mappings on  $\D R$, under the OSC.
Finally, we would like to say that, there are some algorithms to obtain CVTs with $n$-means for any $n\geq 1$ of a region with an absolutely continuous probability measure (see \cite{J}, and the references therein); but to the best of our knowledge there is no such algorithm  for a singular continuous probability measure. So, the result, in this paper, is the first advance in this direction.

\section{Basic definitions and lemmas}

By a \textit{string} or a \textit{word} $\sigma$ over an alphabet $\{1, 2\}$, we mean a finite sequence $\gs:=\gs_1\gs_2\cdots \gs_k$
of symbols from the alphabet, where $k\geq 1$, and $k$ is called the length of the word $\gs$.  A word of length zero is called the \textit{empty word}, and is denoted by $\emptyset$.  By $\{1, 2\}^*$ we denote the set of all words
over the alphabet $\{1, 2\}$ of some finite length $k$ including the empty word $\emptyset$. For any two words $\gs:=\gs_1\gs_2\cdots \gs_k$ and
$\tau:=\tau_1\tau_2\cdots \tau_\ell$ in $\{1, 2\}^*$, by
$\gs\tau:=\gs_1\cdots \gs_k\tau_1\cdots \tau_\ell$ we mean the word obtained from the
concatenation of the two words $\gs$ and $\tau$. Let $S_1$ and $S_2$ be two contractive similarity mappings on $\D R$ given by $S_1(x)=r_1 x$ and $S_2(x)=r_2 x+(1-r_2)$, where $0<r_1, r_2<1$ and $r_1+r_2<1$. Let $(p_1, p_2)$ be a probability vector with $0<p_1, p_2<1$ and $p_1+p_2=1$. For $\gs:=\gs_1\gs_2 \cdots\gs_k \in \{ 1, 2\}^k$, set $S_\gs=S_{\gs_1}\circ \cdots \circ S_{\gs_k}$, $s_\gs=s_{\gs_1}s_{\gs_2}\cdots s_{\gs_k}$, $p_\gs=p_{\gs_1}p_{\gs_2}\cdots p_{\gs_k}$,
and $J_\gs=S_{\gs}([0, 1])$. For the empty word $\es$, by $S_\es$ we mean the identity mapping on $\D R$, and we write $J_\es=S_\es([0,1])=[0, 1]$. Then the set $C=\bigcap_{k\in \mathbb N} \bigcup_{\gs \in \{1, 2\}^k} J_\gs$ is known as the \textit{Cantor set} generated by the two mappings $S_1$ and $S_2$, and equals the support of the probability measure $P$ given by $P=p_1 P\circ S_1^{-1}+p_2 P\circ S_2^{-1}$. For $\gs \in \set{1, 2}^k$, $k\geq 1$, the intervals $J_{\sigma  1}$, $J_{\sigma  2}$ into which $J_\sigma$
is split up at the $(k+1)$th level are called the children of $J_\sigma$.

Let us now give the following two lemmas.

\begin{lemma} \label{lemma1}
Let $f : \mathbb R \to \mathbb R^+$ be Borel measurable and $k\in \mathbb N$. Then
\[\int f dP=\sum_{\sigma \in \{1, 2\}^k} p_\gs \int f \circ S_\sigma dP.\]
\end{lemma}

\begin{proof}
We know $P=p_1  P\circ S_1^{-1} +p_2 P\circ S_2^{-1}$, and so by induction $P=\sum_{\sigma \in \{1, 2\}^k}p_\gs P\circ S_\sigma^{-1}$, and thus the lemma is yielded.
\end{proof}

\begin{lemma} \label{lemma2} Let $X$ be a random variable with probability distribution $P$. Then, the expectation $E(X)$ and the variance $V:=V(X)$ of the random variable $X$ are given by
\[E(X)=\frac {p_2(1-r_2)}{1-p_1r_1-p_2r_2}  \text{ and } V=-\frac{p_2 \left(r_1-1\right){}^2 \left(\left(p_1 r_1-1\right){}^2+p_2 \left(p_1 r_1^2-1\right)\right)}{\left(p_1 r_1+p_2 r_2-1\right){}^2 \left(p_1 r_1^2+p_2 r_2^2-1\right)}.\]
\end{lemma}
\begin{proof} Using Lemma~\ref{lemma1}, we have
\begin{align*}
E(X)&=\int x dP(x)= p_1\int S_1(x) dP(x) +p_2  \int S_2(x) dP(x)\\
=p_1 \int r_1x d P(x)&+p_2\int (r_2 x+1-r_2) dP(x)=p_1 r_1 E(X)+p_2r_2 E(X)+(1-r_2)p_2,
\end{align*}
which implies $E(X)=\frac{p_2(1-r_2)}{1-p_1r_1-p_2r_2}$. Now,
\begin{align*}
&E(X^2)=\int x^2 dP(x)=p_1 \int x^2 dP\circ S_1^{-1}(x)+p_2 \int x^2 dP\circ S_2^{-1}(x) \\
&=p_1 \int r_1^2x^2 dP(x) + p_2 \int (r_2 x+(1-r_2))^2 dP(x)
\end{align*}
which after simplification yields, $E(X^2)=\frac{p_2 \left(r_2-1\right){}^2 \left(-p_1 r_1+p_2 r_2+1\right)}{\left(p_1 r_1+p_2 r_2-1\right) \left(p_1 r_1^2+p_2 r_2^2-1\right)}$, and hence
\[V=E(X-E(X))^2=E(X^2)-\left(E(X)\right)^2 =-\frac{p_2 \left(r_2-1\right){}^2 \left(\left(p_1 r_1-1\right){}^2+p_2 \left(p_1 r_1^2-1\right)\right)}{\left(p_1 r_1+p_2 r_2-1\right){}^2 \left(p_1 r_1^2+p_2 r_2^2-1\right)},\]
which is the lemma.
\end{proof}

The following two notes are in order.
\begin{note}\label{note1}  For any $x_0 \in \D R$, we have
$\int(x-x_0)^2 dP(x) =V(X)+(x_0-E(X))^2$. Thus, one can deduce that the optimal set of one-mean is the expected value and the corresponding quantization error is the variance $V$ of the random variable $X$. For $\sigma \in \{1, 2\}^k$, $k\geq 1$, using Lemma~\ref{lemma1}, we have
\begin{align*}
&E(X : X \in J_\sigma) =\frac{1}{P(J_\sigma)} \int_{J_\sigma} x  dP(x)=\int_{J_\sigma} x  dP\circ S_\sigma^{-1}(x)=\int S_\sigma(x)  dP(x)=E(S_\sigma(X)).
\end{align*}
Since $S_1$ and $S_2$ are similarity mappings, it is easy to see that $E(S_j(X))=S_j(E(X))$ for $j=1, 2$, and so by induction, $E(S_\sigma(X))=S_\sigma(E(X))$
 for $\sigma\in \{1, 2\}^k$, $k\geq 1$.
\end{note}

\begin{note}

For words $\gb, \gg, \cdots, \gd$ in $\set{1,2}^\ast$, by $a(\gb, \gg, \cdots, \gd)$ we denote the conditional expectation of the random variable $X$ given $J_\gb\uu J_\gg \uu\cdots \uu J_\gd,$ i.e.,
\begin{equation} \label{eq1} a(\gb, \gg, \cdots, \gd)=E(X|X\in J_\gb \uu J_\gg \uu \cdots \uu J_\gd)=\frac{1}{P(J_\gb\uu \cdots \uu J_\gd)}\int_{J_\gb\uu \cdots \uu J_\gd} x dP.
\end{equation}
Thus, by Note~\ref{note1}, $a(\gs)=S_\sigma(E(X))$
 for $\sigma\in \{1, 2\}^\ast$. Moreover, for any $x_0\in \D R$ and $\gs \in \set{1, 2}^\ast$, we have
\begin{equation} \label{eq2} \int_{J_\gs}(x-x_0)^2 dP= p_\gs \int (x -x_0)^2 dP\circ S_\gs^{-1}=p_\gs  \Big(s_\gs^2V+(S_\gs(E(X))-x_0)^2\Big).\end{equation}
The expressions \eqref{eq1} and \eqref{eq2} are useful to obtain the CVTs and the corresponding distortion errors with respect to the probability distribution $P$.
\end{note}
In the next section, we give the algorithm which is the main result of the paper.

\section{Algorithm to determine the CVTs with $n$-means for any $n\geq 1$}
In this section, first we give an algorithm to obtain the centroidal Voronoi tessellations with $n$-means for any $n\geq 1$ of the Cantor set $C$ supported by the probability measure $P$ defined in the previous section. To run the algorithm one can code it either in Mathematica, Matlab, C++ or in any other programming language. To write the algorithm, let us identify any word $\gs_1\gs_2\cdots \gs_k \in \set{1, 2}^k$, $k\geq 1$, by $\set{\gs_1, \gs_2, \cdots,  \gs_k}$.
For any positive integer $m$ denote the words in the set $\set{1, 2}^m$ by the indices $1, 2, \cdots, 2^m$ in increasing order,  that is, for any $i, j \in \set{1, 2, \cdots, 2^m}$, by $i<j$, it is meant $S_i(x)<S_j(x)$ for $x\in \D R$.  For indices $i, j \in \set{1, 2, \cdots, 2^m}$, $i\leq j$, by $[i, j]$, we mean the block which contains all the words with indices from $i$ to $j$; and by $[i, i]$, it is meant the word with index $i$.  By $a[i, j]$, it is meant the expected value, as defined in \eqref{eq1}, of the random variable $X$ with distribution $P$ taking values on $J_\gs$ for some $\gs \in [i, j]$. For example, if $m=3$, then
\begin{align*}
\set{1, 2}^3&=\set{\set{1, 1, 1}, \set{1, 1, 2}, \set{1, 2, 1}, \set{1, 2, 2}, \set{2, 1, 1}, \set{2, 1, 2}, \set{2,
   2, 1}, \set{2, 2, 2}} \\
   &=\set{1, 2, \cdots, 8}.
\end{align*}
Thus, here $1=\set{1, 1, 1}, \, 2=\set{1, 1, 2}, \, \cdots, \, 8=\set{2, 2, 2}$. As $S_{\set{\gs_1, \gs_2, \cdots,  \gs_k}}$ is identical with $S_{\gs_1\gs_2\cdots \gs_k}$, for any $i, j\in \set{1, 2, \cdots, 8}$ with $i<j$, one can see that $S_i(x)<S_j(x)$ for $x\in \D R$.
Let us now state the algorithm as follows:
\subsection{Algorithm}

$(i)$ Choose an initial positive integer $m$ so that $n\leq 2^m$.

$(ii)$ Partition the set $\set{1, 2}^m$ into $n$ blocks $[i_\ell+1, i_{\ell+1}]$ for $\ell=0, 1, \cdots, n-1$  in all possible ways, where $i_0=0$ and $i_n=2^m$.  For each partition obtained in Step $(ii)$, check if $S_{i_{\ell+1}}(1)\leq \frac 12 (a[i_\ell+1, i_{\ell+1}]+a[i_{\ell+1}+1, i_{\ell+2}])\leq S_{i_{\ell+1}+1}(0)$ for all $\ell=0, 1, \cdots, n-2$; if so, then the $n$ blocks $[i_\ell+1, i_{\ell+1}]$ in the partition form a centroidal Voronoi tessellation, $P$-almost surely,  with $n$-centroids $a[i_\ell+1, i_{\ell+1}]$ for $\ell=0, 1, \cdots, n-1$. If at least one set of $n$-centroids is obtained, terminate, otherwise, go to step $(iii)$.

$(iii)$ Replace $m$ by $m+1$ and return to Step $(ii)$.

\begin{note} Let $C(n, 2^m)$ be the collection of all the sets of $n$-centroids obtained after the completion of one cycle of the algorithm for some positive integer $m$ with $n\leq 2^m$, then it is easy to see that $C(n, 2^m) \ci C(n, 2^{m+1})$.
Once a set of $n$-centroids are known the corresponding Voronoi tessellation can easily be obtained. Thus, in the sequel, sometimes we will identify a Voronoi tessellation by the set of its centroids.  By using the formula \eqref{eq2}, one can also obtain the distortion error for each Voronoi tessellation.
\end{note}

\begin{note} Substantial efficiency gains may be made by checking the conditions in step (ii) as early as possible.  For example, after choosing $i_1$ and $i_2$, we have enough information to check the $\ell=0$ condition; if this fails we need not enumerate any of the other $i_\ell$.  Furthermore, $a[i_1,i_2]$ is monotonic in $i_2$ which means that more extreme $i_2$ may be ruled out immediately.
\end{note}

\begin{note}
When $r_1=r_2=r$ and $p_1=p_2=1/2$, the measure is reflection symmetric about $x=1/2$.  This symmetry may be taken into account by assuming (without loss of generality) that the $\ell$th partition element is larger than or the same size as the $(n-1-\ell)$th partition element, with $\ell=\lfloor n/2\rfloor-1$.\label{n:3.3}
\end{note}

{\bf Algorithm step (ii) in more detail.} Taking into account Note~\ref{n:3.3} we can write an efficient form of step (ii) as follows:

$(iia)$ Enumerate $i_1$, calculate $a[1,i_1]$, and set $\ell=2$.

$(iib)$ Set upper and lower bounds for $i_\ell$, initially at $L_\ell=i_{\ell-1}+1$ and $U_\ell=2^m$ respectively.  Enumerate $i_\ell$ ignoring values outside the bounds.
If $\ell=n$ there is a single value $i_n=2^m$.

$(iic)$ Symmetry condition (if applicable): if $\ell=n-\lfloor n/2\rfloor+1$ and the block $[i_{\ell-1}+1,i_\ell]$ is larger than the block $[i_{n-\ell}+1,i_{n-\ell+1}]$, go to $(iig)$.

$(iid)$ Inequalities: Calculate $a[i_{\ell-1}+1,i_\ell]$.  Check if $S_{i_{\ell-1}}(1)\leq \frac 12 (a[i_{\ell-2}, i_{\ell-1}]+a[i_{\ell-1}+1, i_{\ell}])\leq S_{i_{\ell-1}+1}(0)$.  If the first inequality is violated, reset the lower bound: $L_\ell=i_\ell$.   If the second is violated, reset the upper bound: $U_\ell=i_\ell$.  If either inequality is violated, continue with the enumeration of $i_\ell$ and return to $(iic)$.

$(iie)$ Increment $\ell$.  If $\ell\leq n$, return to $(iib)$.

$(iif)$ Print the centroidal Voronoi tessellation.

$(iig)$ Continue with the pending enumerations.

\section{Exact examples}
Let us now give the following examples.
\begin{example} Let $r_1=r_2=\frac 13$. Then, the Cantor set defined in the previous section reduces to the classical Cantor set generated by the two mappings $S_1, S_2$ given by $S_1(x)=\frac 13 x$ and $S_2(x)=\frac 13 x+\frac 23$, and is supported by the probability measure $P$ given by $P=\frac 12 P\circ S_1^{-1}+\frac 12 P\circ S_2^{-1}$.
 \begin{defi} \label{def1} For $n\in \D N$ with $n\geq 2$ let $\ell(n)$ be the unique natural number with $2^{\ell(n)} \leq n<2^{\ell(n)+1}$. For $I\sci \set{1, 2}^{\ell(n)}$ with card$(I)=n-2^{\ell(n)}$ let $\gb_n(I)$ be the set consisting of all midpoints $a_\gs$ of intervals $J_\gs$ with $\gs \in \set{1,2}^{\ell(n)} \setminus I$ and all midpoints $a_{\gs 1}$, $a_{\gs 2}$ of the children of $J_\gs$ with $\gs \in I$. Formally,
\[\gb_n(I)=\set{a_\gs : \gs \in \set{1,2}^{\ell(n)} \setminus I} \uu \set{a_{\gs 1} : \gs \in I} \uu \set {a_{\gs 2} : \gs \in I}.\]
\end{defi}
In \cite{GL2} it was shown that $\gb_n(I)$ forms an optimal set of $n$-means for any $n\geq 2$.  Let $\gb_n$ denote all the optimal sets of $n$-means in this case. Then,
\[\gb_3=\Big\{\{\frac{1}{18},\frac{5}{18},\frac{5}{6}\}, \{\frac{1}{6},\frac{13}{18},\frac{17}{18}\}\Big\}.\]
 Notice that $\set{1, 2}^2=\set{\set{1, 1}, \set{1, 2}, \set{2, 1}, \set{2, 2}}$. Now, for $n=3$ and $m=2$, if the algorithm is run after completion of one cycle, one can see that there are two sets of three-centroids occur: one for the blocks  $[1, 1]$, $[2, 2],$ $[3, 4]$; and one for the blocks $[1, 2]$, $[3, 3]$ and $[4, 4]$. Thus, the two CVTs in this case are $\set{a[1, 1], a[2, 2], a[3, 4]}$ and $\set{a[1, 2], a[3, 3], a[4, 4]}$ which form $\gb_3$.
Now, if one runs the algorithm for the second time, that is when $n=3$ and $m=3$, then as
\[\set {1, 2}^3=\set{\set{1,1,1}, \set{1,1,2}, \set{1,2,1}, \set{1, 2, 2}, \set{2, 1, 1},\set{2, 1,2}, \set{2, 2, 1}, \set{2, 2, 2} }\] the two sets of three-centroids occur as follows: one for the blocks  $[1, 2]$, $[3, 4]$ $[5, 8]$; and one for the blocks $[1, 4]$, $[5, 6]$ and $[7, 8]$. Thus, the two CVTs in this case are
\[ \set{a[1, 2], a[3, 4], a[5, 8]} \te{ and } \set{a[1, 4], a[5, 6], a[7, 8]} \] which form $\gb_3$.
Similarly, by running the algorithm for the third time and fourth time, one can see that $C(3, 2^4)$ consists of the sets $\set{a[1,4], a[5, 8], a[9, 16]}$ and $\set{a[1, 8], a[9, 12], a[13, 16]}$ which is $\gb_3$; but, $C(3, 2^5)$ consists of the following three sets:
\[\set{a[1, 8], a[9, 16], a[17, 32]}, \ \set{a[1, 16], a[17, 24], a[25, 32]}, \ \te{ and } \set{a[1, 15], a[16, 17], a[18, 32]}\] among which $\set{a[1, 8], a[9, 16], a[17, 32]}$ and $\set{a[1, 16], a[17, 24], a[25, 32]}$ form $\gb_3$. Thus, we see that
\[\gb_3=C(3, 2^2)=C(3, 2^3)=C(3, 2^4)\sci C(3, 2^5).\]
Again, by the above definition, we have
\[\gb_4=\Big\{\frac 1 {18}, \frac{5}{18}, \frac{13}{18}, \frac {17}{18}\Big \}.\]
Now, if we start running the algorithm, taking $n=4, m=2$, one can see that $C(4, 2^2)$ consists of only one set $\set{a[1, 1], a[2, 2], a[3, 3], a[4, 4]}$ which is $\gb_4$.

$C(4, 2^3)$ consists of the following sets:
\begin{align*} &\set{a[1,1], a[2, 2], a[3, 4], a[5, 8] }, \, \set{a[1, 2], a[3, 3], a[4, 4], a[5, 8]}, \\
& \set {a[1, 2], a[3, 4], a[5, 6], a[7, 8]}, \,  \set {a[1, 4], a[5, 5], a[6, 6], a[7, 8]}, \set {a[1, 4], a[5, 6], a[7, 7], a[8, 8]}
\end{align*}
among which $\set {a[1, 2], a[3, 4], a[5, 6], a[7, 8]}=\gb_4$.

$C(4, 2^4)$ consists of the following sets:
\begin{align*} &\set{a[1,2], a[3, 4], a[5, 8], a[9, 16] }, \, \set{a[1, 4], a[5, 6], a[7, 8], a[9, 16]}, \\
& \set {a[1, 4], a[5, 8], a[9, 12], a[13, 16]}, \,  \set {a[1, 8], a[9, 10], a[11, 12], a[13, 16]}, \\
&\set {a[1, 8], a[9, 12], a[13, 14], a[15, 16]}
\end{align*}
among which $\set {a[1, 4], a[5, 8], a[9, 12], a[13, 16]}=\gb_4$.

$C(4, 2^5)$ consists of the following sets:
\begin{align*} &\set{a[1, 4], a[5, 8], a[9, 16], a[17, 32] }, \, \set{a[1, 8], a[9, 12], a[13, 16], a[17, 32]}, \\
& \set {a[1, 8], a[9, 16], a[17, 24], a[25, 32]}, \,  \set {a[1, 16], a[17, 20], a[21, 24], a[25, 32]}, \\
&\set {a[1, 16], a[17, 24], a[25, 28], a[29, 32]}
\end{align*}
among which $\set {a[1, 8], a[9, 16], a[17, 24], a[25, 32]}=\gb_4$.
Thus, one can see that $\gb_4=C(4, 2^2) \sci C(4, 2^3)=C(4, 2^4)=C(4, 2^5)$.
 \end{example}

\begin{remark} Recall that an optimal set of $n$-means forms a CVT with $n$-means; however, the converse is not always true, which is also verified from the above example. The following example shows that if one runs the algorithm for some $n$ and $m$ with $n\leq 2^m$, initially there can be no output.
\end{remark}

\begin{example} In the Cantor set construction in Section 2, let us take $r_1=r_2=\frac 49$ and $P=\frac 12 P\circ S_1^{-1}+\frac 12 P\circ S_2^{-1}$.

Now, if we keep running our algorithm for $n=3$ starting with $m=2$, then we see that both
$C(3, 2^2)$ and $C(3, 2^3)$ are empty sets, that is, there is no output for $m=2$ and $m=3$. On the other hand,
$C(3, 2^4)$ consists of the sets $\set{a[1, 4], a[5, 9], a[10, 16]}$ and $\set{a[1, 7], a[8, 12], a[13, 16]}$ which are, respectively,
$\set{0.0987654, 0.391556, 0.806737}$ and $\set{0.193263, 0.608444, 0.901235}$.

$C(3, 2^5)$ consists of the sets
$\set{a[1, 8], a[9, 18], a[19, 32]}$, $\set{a[1, 11], a[12, 20], a[21, 32]}$, \\ $\set{a[1, 12], a[13, 21], a[22, 32]}$,
and $ \set{a[1, 14], a[15, 24], a[25, 32]}$, which are, respectively, \\
$\set{0.0987654, 0.391556, 0.806737}$,
$\set{0.147939, 0.48067,0.83722}$, $\set{0.16278, 0.51933, 0.852061}$ and $\set{0.193263,0.608444,0.901235}$.

Thus, we have $C(3, 2^2)=C(3, 2^3)=\es$, $C(3, 2^4)\neq \es$ and $ C(3, 2^4)\sci C(3, 2^5)$.
\end{example}

\begin{remark}

In \cite{R1}, the second author  determined a CVT and the corresponding distortion error for the probability measure $P$ given by $P=\frac 1 2 P\circ S_1^{-1}+\frac 12 P\circ S_2^{-1}$ which has support the Cantor set generated by $S_1(x)=r x$ and $S_2(x)=rx+(1-r)$, where $0.4364590141\leq r\leq 0.4512271429$.
There it was also shown that if $0.4371985206<r\leq 0.4384471872$ and $n$ is not of the form $2^{\ell(n)}$ for any positive integer $\ell(n)$, then the distortion error of the CVT obtained using the formula given in \cite{R1} is smaller than the distortion error of the CVT obtained using the formula given by Graf and Luschgy in \cite{GL2}. In the following example we show that the CVT obtained using the formula given in \cite{R1}, though gives smaller distortion error, is not an optimal CVT.
\end{remark}
\begin{example} \label{ex1}
In the construction of the Cantor set, let us take $r_1=r_2=r=0.4375$ which lies in the range $0.4371985206<r\leq 0.4384471872$, and $P=\frac 12 P\circ S_1^{-1}+\frac 12 P\circ S_2^{-1}$.

Now, if we keep running our algorithm for $n=3$ starting with $m=2$, then we see that
\[C(3, 2^2)=\set{\set{a[1, 1], a[2, 2], a[3, 4]}, \, \set{a[1, 2], a[3, 3], a[4, 4]}}\]
which consists of the sets $\set{0.0957031, 0.341797, 0.78125}$ and $\set{0.21875,0.658203,0.904297}$ respectively, and each have the same distortion error $0.0111543$.

If $n=3$, $m=3$, we have
\[C(3, 2^3)=\set{\set{a[1, 2], a[3, 4], a[5, 8]}, \, \set{a[1, 3], a[4, 5], a[6, 8]}, \, \set{a[1, 4], a[5, 6], a[7, 8]}}, \]
i.e., $C(3, 2^3)$ consists of the sets $\set{0.0957031, 0.341797, 0.78125}$, $\set{0.15979, 0.5, 0.84021}$ and the set $\set{0.21875,0.658203,0.904297}$ with the distortion errors $0.0111543$, $0.011019$, and $0.0111543$ respectively.

If $n=3$ and $m=4$, then $C(3, 2^4)$ consists of the sets
\begin{align*}
\set{a[1, 4], a[5, 8], a[9, 16]}&=\{0.0957031, 0.341797, 0.78125\} \te{ with distortion error } 0.0111543, \\
\set{a[1, 4], a[5, 9], a[10, 16]}&=\{0.0957031, 0.389601, 0.809883\} \te{ with distortion error } 0.0111413, \\
\set{a[1, 6], a[7, 10], a[11, 16]}&=\{0.15979, 0.5, 0.84021\} \te{ with distortion error } 0.011019,\\
\set{a[1, 7], a[8, 12], a[13, 16]}&=\{0.190117,0.610399,0.904297\} \te{ with distortion error } 0.0111413,\\
\set{a[1, 8], a[9, 12], a[13, 16]}&=\{0.21875,0.658203,0.904297\} \te{ with distortion error } 0.0111543.\\
\end{align*}
\end{example}
If $n=3$ and $m=5$, then $C(3, 2^5)$ consists of the sets
\begin{align*}
\set{a[1, 8], a[9, 16], a[17, 32]}&=\{0.0957031,0.341797,0.78125\} \te{ with distortion error } 0.0111543, \\
\set{a[1, 8], a[9, 17], a[18, 32]}&=\{0.0957031,0.36721,0.795299\} \te{ with distortion error } 0.011127, \\
\set{a[1, 8], a[9, 18], a[19, 32]}&=\{0.0957031,0.389601,0.809883\} \te{ with distortion error } 0.0111413,\\
\set{a[1, 11], a[12, 20], a[21, 32]}&=\{0.14506,0.480202,0.84021\} \te{ with distortion error } 0.0110059,\\
\set{a[1, 12], a[13, 20], a[21, 32]}&=\{0.15979,0.5,0.84021\} \te{ with distortion error } 0.011019,\\
\set{a[1, 12], a[13, 21], a[22, 32]}&=\{0.15979,0.519798,0.85494\} \te{ with distortion error } 0.0110059,\\
\set{a[1, 14], a[15, 24], a[25, 32]}&=\{0.190117,0.610399,0.904297\} \te{ with distortion error } 0.0111413,\\
\set{a[1, 15], a[16, 24], a[25, 32]}&=\{0.204701,0.63279,0.904297\} \te{ with distortion error } 0.011127,\\
\set{a[1, 16], a[17, 24], a[25, 32]}&=\{0.21875,0.658203,0.904297\} \te{ with distortion error } 0.0111543.\\
\end{align*}

Now, we give the following observations:

$(i)$ In $C(3, 2^4)$ we have five CVTs, among which the CVTs $\set{a[1, 4], a[5, 9], a[10, 16]}$ and $\set{a[1, 7], a[8, 12], a[13, 16]}$ are the two CVTs with three-means that were obtained using the formula given in \cite{R1}. Moreover, among all the CVTs in $C(3, 2^4)$ the CVT $\set{a[1, 6], a[7, 10], a[11, 16]}$ gives the smallest distortion error, even smaller than the distortion error of any of the CVTs with three-means obtained in \cite{R1}.

$(ii)$ In $C(3, 2^5)$ we obtain nine CVTs, among which the CVTs $\set{a[1, 8], a[9, 18], a[19, 32]}$ and $\set{a[1, 14], a[15, 24], a[25, 32]}$ are the two CVTs with three-means that were obtained using the formula given in \cite{R1}. There are five CVTs in $C(3, 2^5)$ which have smaller distortion errors than the distortion error of any of the CVTs with three-means obtained using the formula in \cite{R1}. In addition, in $C(3, 2^5)$ we obtained two new CVTs which are $\set{a[1, 11], a[12, 20], a[21, 32]}$ and $\set{a[1, 12], a[13, 21], a[22, 32]}$, and have the smallest distortion error among all the distortion errors of the CVTs obtained in $C(3, 2^5)$.

\begin{remark}
By observations $(i)$ and $(ii)$ in Example~\ref{ex1}, we can say that for $0.4371985206<r\leq 0.4384471872$, the CVTs obtained using the formula given in \cite{R1} are not optimal.
\end{remark}

\begin{example} \label{example2} In the Cantor set construction, let us take $r_1=\frac 14$, $r_2=\frac 12$ and $P=\frac 1 4 P\circ S_1^{-1}+\frac 3 4 P\circ S_2^{-1}$, i.e., $p_1=\frac 1 4$ and $p_2=\frac 34$.

Now, for $n=3$ and $m=2$ if we run the algorithm, one can see:
$C(3, 2^2)$ consists of only one set $\set{a[1, 2], a[3, 3], a[4, 4]}=\{0.166667,0.583333,0.916667\}$ with distortion error $ 0.00561683.$

$C(3, 2^3)$ consists of the sets
\begin{align*}
\set{a[1, 4], a[5, 6], a[7, 8]}&=\{0.166667,0.583333,0.916667\} \te{ with distortion error } 0.00561683,\\
\set{a[1, 4], a[5, 7], a[8, 8]}&=\{0.166667,0.672619,0.958333\} \te{ with distortion error } 0.00617487.
\end{align*}

$C(3, 2^4)$ consists of the sets
\begin{align*}
\set{a[1, 8], a[9, 12], a[13, 16]}&=\{0.166667,0.583333,0.916667\} \te{ with distortion error } 0.00561683,\\
\set{a[1, 8], a[9, 13], a[14, 16]}&=\{0.166667,0.611294,0.927083\} \te{ with distortion error } 0.00562968,\\
\set{a[1, 8], a[9, 14], a[15, 16]}&=\{0.166667,0.672619,0.958333\} \te{ with distortion error } 0.00617487.
\end{align*}
If we fix $n=3$ and keep running the algorithm, one can see:
$\{0.166667,0.583333,0.916667\}$ is the only CVT with smallest distortion error $0.00561683$. In fact in \cite{R5}, it was shown that $\{0.166667,0.583333,0.916667\}$ is the only optimal set of three-means with quantization error $0.00561683$.

If we put $n=4$, $m=3$ and run the algorithm, we can see:
$C(4, 2^3)$ consists of the following sets,
\begin{align*}
&\{0.0416667,0.208333,0.583333,0.916667\} \te{ with distortion error } 0.00431475,\\
&\{0.0416667,0.208333,0.672619,0.958333\} \te{ with distortion error } 0.00487278,\\
&\{0.0863095,0.229167,0.583333,0.916667\} \te{ with distortion error } 0.00436125,\\
&\{0.0863095,0.229167,0.672619,0.958333\} \te{ with distortion error } 0.00491929,\\
&\{0.166667,0.583333,0.791667,0.958333\} \te{ with distortion error } 0.00268714,
\end{align*}
among which the set $\{0.166667,0.583333,0.791667,0.958333\}$ has the smallest distortion error. In fact, as shown in \cite{R5}, it is the only optimal set of four-means with quantization error  0.00268714. Thus, for a fixed $n$ by running our algorithm, if needed for several times, one can see that the CVTs with smallest distortion error obtained in our case is actually the optimal sets of $n$-means as obtained by L. Roychowdhury (see \cite{R5}).
\end{example}

\begin{figure}
\centerline{\includegraphics[width=650pt]{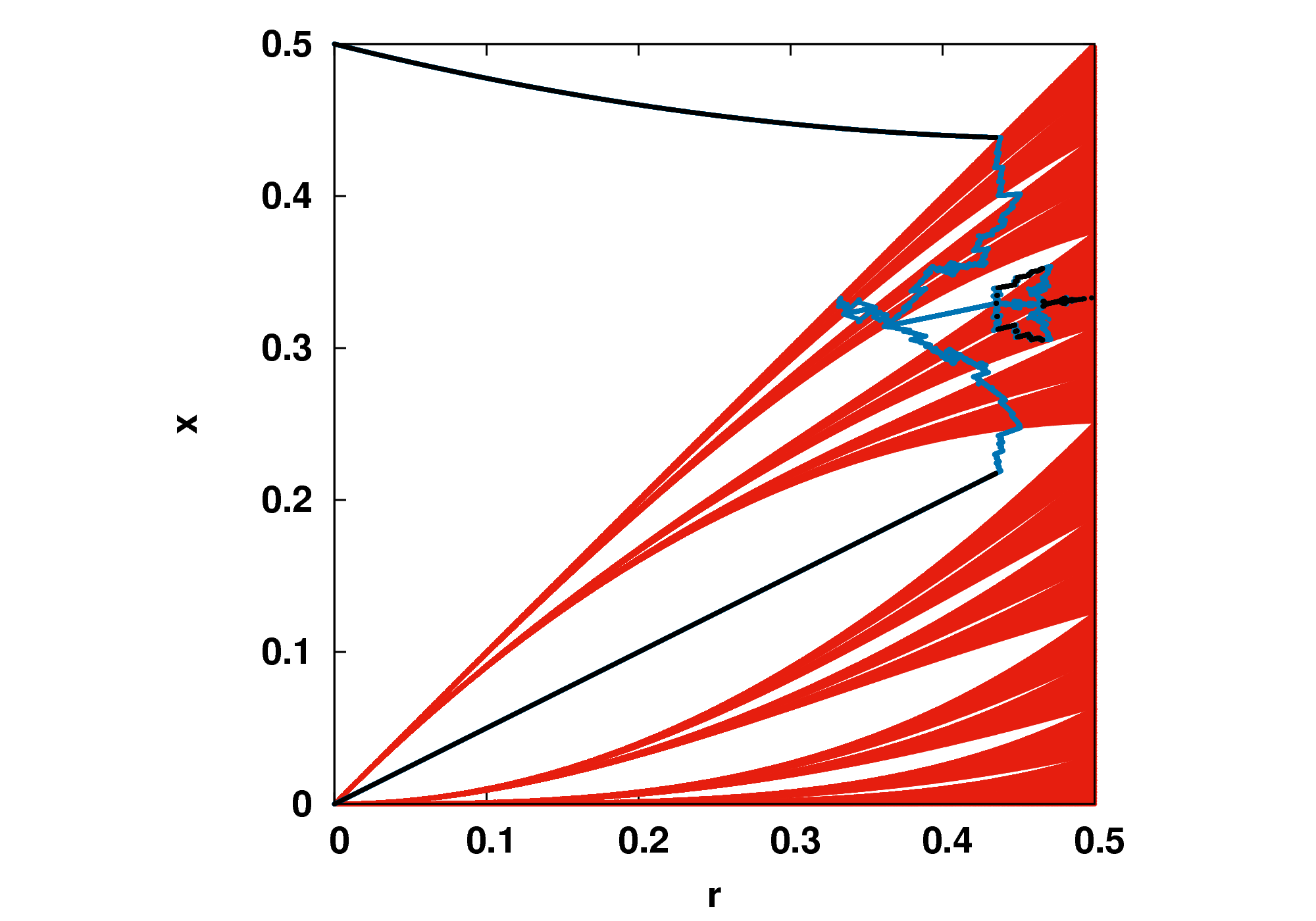}}%\includegraphics[width=280pt]{lightzoom.png}}
\caption{\label{f:lightning} Half of the symmetric Cantor set (red), showing CVTs (blue) and optimal CVTs (black).  The quantity plotted is the location of the boundary between the first and second partition elements.  Two values are optimal for almost all $r$ due to the symmetry of the Cantor set.}  %Right: The same plot, zoomed in.}
\end{figure}

As our final example, we perform numerical simulations to elicit the dependence on the scale factor. We now consider the uniform symmetric Cantor measure with $r_1=r_2=r$ and $p_1=p_2=1/2$.  The smallest nontrivial partition is $n=3$.  We run the algorithm to a level $m=14$.  The results are shown in Figure~\ref{f:lightning}.  For $r<0.43$ we see that the optimal partition is, as shown in previous literature~\cite{GL2} for $r=1/3$, given by splitting the set into one half and two quarters.  For $r>0.33$ there is an intricate structure of CVTs, some of which are optimal for $r>0.43$.  %Very close to $r=0.5$, the results are limited by the chosen level.
Even in this simplest situation ($n=3$ for the uniform symmetric Cantor set), a complete analysis is far from trivial, and we postpone it to future publications.

Let us now conclude the paper with the following remark.

\begin{remark}
The algorithm given in this paper can be used to obtain the CVTs with $n$-generators, $n\geq 1$, for any singular continuous probability measure on $\D R$ supported by a Cantor like set defined as follows:

Let $(n_k)$ be a bounded sequence of positive integers such that $n_k\geq 2$ for all $k\geq 1$. Let $S_{kj}$, $1\leq j\leq n_k$, $k\geq 1$, be contractive similarity mappings on $\D R$ satisfying the open set condition with contractive ratios $0<c_{kj}<1$ such that $\sum_{j=1}^{n_k} c_{kj}<1$.  Let $p_{kj}$ be the probabilities associated with $S_{kj}$ such that $0<p_{kj}<1$  and $\sum_{j=1}^{n_k} p_{kj}=1$ for all $k\geq 1$. Let $W_n:=\prod_{k=1}^{n}\set{1, 2, \cdots, n_k}$. Then, by the set of all words $W^\ast$ it is meant:
$W^\ast=\UU_{n=1}^{\infty} W_n$.
Let $P$ be the probability measure supported by the limit set generated by the contractive mappings $S_{kj}$ on $\D R$ associated with the probabilities $p_{kj}$. Then, it is well-known that $P$ is the image measure of the product measure $\hat P$ on the space $\prod_{k=1}^\infty\set{1, 2, \cdots, n_k}$, where $\hat P=\prod_{k=1}^\infty(p_{k1}, p_{k2}, \cdots, p_{kn_k})$, under a coding map $\pi$.  For such a probability distribution $P$, our algorithm also works to determine the CVTs with $n$-means for any $n\geq 1$ with the following changes to be made: Replace $n\leq 2^m$ and $\set{1, 2}^m$ in the algorithm, respectively,  by $n\leq \prod_{k=1}^m n_k$ and $\prod_{k=1}^m \set{1, 2, \cdots, n_k}$.
\end{remark}

\end{document}